\documentclass[10pt]{amsart}
\usepackage{amsmath}
  \usepackage{graphics} 
\usepackage{graphicx}  
  \usepackage{epsfig} 
 \usepackage[colorlinks=true]{hyperref}
  \usepackage{float}
  \usepackage{caption}
\usepackage{amsmath}
\usepackage{wrapfig}
\usepackage{graphicx}
\usepackage{subcaption}
\usepackage{amssymb}

\newtheorem{theorem}{Theorem}[section]
\newtheorem{corollary}{Corollary}

\newtheorem{lemma}[theorem]{Lemma}

\newtheorem{conjecture}{Conjecture}

\theoremstyle{definition}

\newtheorem{remark}{Remark}


\title[ additional food ] 
    {A Constant Introduction of Additional Food Does Not Yield Pest Eradication in Finite Time}

\begin{document}
\maketitle

\centerline{\scshape   Rana D. Parshad$^{1}$, Sureni Wickramsooriya$^{2}$ and Susan Bailey$^{3}$  }
\medskip
{\footnotesize

 \centerline{ 1) Department of Mathematics,}
 \centerline{Iowa State University,}
   \centerline{Ames, IA 50011, USA.}
   \medskip
 \centerline{2) Department of Mathematics,}
 \centerline{Clarkson University,}
   \centerline{ Potsdam,  NY 13699, USA.}
   \medskip
 \centerline{3) Department of Biology,}
 \centerline{Clarkson University,}
   \centerline{ Potsdam,  NY 13699, USA.}  
   
 }

\begin{abstract}
Biological control, the use of predators and pathogens to control target pests, is a promising alternative to chemical control. It is hypothesized that the introduced predators efficacy can be boosted by providing them with an additional food source. The current literature \cite{SP11} claims that if the additional food is of sufficient constant quantity and quality then pest eradication is possible in \emph{finite} time. 
The purpose of the current manuscript is to show that to the contrary, pest eradication is \emph{not possible} in finite time, for \emph{any} quantity and quality of constant additional food. 
However for a density dependent quantity of additional food, we show pest eradication in finite time is indeed possible. Our results have large scale implications for the effective design of biological control methods involving additional food.
\end{abstract}

 \section{Introduction}
\subsection{Overview}
Biological pests or invasive species present a significant and increasing threat to economies, ecosystems, and human health worldwide \cite{PZM05, PSCEWT16, L16, E11}. To combat these unwanted pests, chemical pesticides are often used. However, pesticide application can be expensive, require many repeat applications before it is effective, and can have unintended toxic effects on other species the environment \cite{PB14}. 
A promising alternative to chemical pesticides is biological control, where a second species is introduced to increase mortality, or at least slow the growth rate, of the target pest species. This antagonist species can play a number of roles including parasite, pathogen, predator, or competitor \cite{V96, C15, B07, K17}. Here we focus on biological control via predation. 
While biocontrol via predation has proven effective in some cases, in most cases predators are not able to completely eradicate the target pest species \cite{V96}. One approach to try to boost predator efficiency is to provide them an additional food source. 

Food supplementation has been tested in a range of field studies, sometimes leading to significant reductions in numbers of pests (pests declined in 28 of the 59 trials reviewed in \cite{W08}), but rarely is complete elimination reported. There has been some investigation into factors that may determine the efficacy of the food supplementation. In some studies, increased concentration of the food supplement has been shown to increase predator numbers and decrease pests \cite{CO98, ES93, SC07, S10, W16, T15, R95, T15}. Other studies show that changing the quality or the nutritional makeup of the additional food can have different impacts on pest numbers (e.g. spraying with either carbohydrates, protein, or both; \cite{ES93}). Interestingly, the duration and frequency does not appear to have a consistant effect on predator or pest populations \cite{W08}. The reason for this is not clear, but understanding these inconsistent outcomes would be an invaluable step towards using natural predators as an effective and more widely applied pest control approach.

Mathematical models that describe predator-pest dynamics are useful for identifying conditions under which pest elimination is expected to be successful or not. Note, a mathematical modeling approach is particularly helpful for exploring the food supplementation approach as it is not intuitively obvious that adding food will actually have the desired effect of more effectively eradicating the pest species. One might expect that in some cases, for example, the predator population might simply eat more of the additional food with no effect on their net consumption of prey, or perhaps even fewer prey might be consumed because some predators switch to exclusively consume the added food source.

\subsection{Prior Work}
 A number of mathematical models that describe predator-pest dynamics with an additional food source have been developed, stemming from three key papers in the literature \cite{SP07, SP10, SP11}. 
 In these works the dynamics of an introduced predator depredating on a target pest is modeled by the following system,
\begin{equation}
\label{Eqn:1}
\frac{dx}{dt} = x(1-\frac{x}{\gamma}) - \frac{xy}{1+\alpha \xi + x}, \    \frac{dy}{dt} =\frac{\beta xy}{1+\alpha \xi + x} + \frac{\beta \xi y}{1+\alpha \xi + x} - \delta y.
\end{equation}

Here $x(t),y(t)$ are the number/density of a pest and predator species, $\gamma$ is the carrying capacity of the pest, $\beta$ is the conversion efficiency of the predator, $\delta$ is the death rate of the predator, $\frac{1}{\alpha}$ is the quality of the additional food provided to the predator and $\xi$ is the quantity of additional food provided to the predator. Note, $\gamma, \beta, \delta, \alpha$ are all positive constant parameters.

In \cite{SP07, SP10} it is claimed that the above model can facilitate pest extinction in finite time. This is however not proved until \cite{SP11}.
We recap the result of interest from \cite{SP11} which quantifies the efficacy of the predator to achieve pest eradication when supplemented with additional food \cite{SP11} via \eqref{Eqn:1},

\begin{lemma}
[Srinivasu $\&$ Prasad, Bulletin of Mathematical Biology, 2011] 

(a) If the quality of the additional food satisfies $\beta - \delta \alpha  > 0$, then prey can be eradicated
from the ecosystem in a finite time by providing the predator with additional
food of quantity $\xi > \frac{\delta}{\beta - \delta \alpha }$.

(b) If the quality of the additional food satisfies $\beta - \delta \alpha  < 0$, then it is not possible to
eradicate prey from the ecosystem through provision of such additional food to the predators.

\label{lem:1}
\end{lemma}

\begin{remark}
Notice if $\xi = 0$, or there is no additional food \eqref{Eqn:1} reduces to the classical Lokta-Volterra predator-prey model, for which we know prey eradication \emph{is not possible}, as the $(0,0)$ state is unstable. Thus the above results are highly promising for the field of biological control, as they show a modification to the classic Lokta-Volterra system, via the introduction of additional food can cause pest eradication. That is a pest free stable state can be reached, after which point the introduced predator can continue to survive on the additional food and 
(a) either  grow in time or (b) reach a steady state.
 Thus the results of \cite{SP11} have led to much research activity recently \cite{SP07, SP10, SP11, CT17, SP18}.
\end{remark}

In the current manuscript,

\begin{enumerate}
\item We show that for a constant quantity $\xi$ of additional food, pest eradication is \emph{not possible} in finite time, even if $\xi > \frac{\delta}{\beta - \delta \alpha }$, and $\beta - \delta \alpha  > 0$. This is shown via Theorem \ref{thm:1a}.

\item Pest eradication in finite time is not possible even if an arbitrarily large quantity of constant additional food is considered, that is even in the limit that $\xi \rightarrow \infty$, or if an arbitrarily high quality of constant additional food is considered, that is even in the limit that $\alpha \rightarrow 0$.
This is shown via Corollary \ref{cor:1bc2}.

\item A constant quantity of additional food can however cause pest eradication in \emph{infinite} time. Decay rates to the extinction state are derived via Lemma \ref{lem:1a}.

\item If the quantity of additional food is not constant but rather pest density dependent, that is $\xi=\xi(x)$, then pest eradication \emph{is possible} in finite time. This is shown via Theorem \ref{thm:1b}.

\item We investigate forms of $\xi=\xi(x)$ that can lead to pest extinction in finite time, and discuss ecological and management consequences of these via Lemma \ref{lem:1bc}, Corollary \ref{cor:1bc1}, and several conjectures that we make in section \ref{conj}.
\end{enumerate}

\section{Finite Time Extinction}
\subsection{Constant Quantity of Additional Food}
We first show Lemma \ref{lem:1} is not quite accurate. That is for a constant quantity of additional food, pest extinction does not occur in \emph{finite} time. We state this via the following Theorem,

\begin{theorem}
\label{thm:1a}
Consider the predator-pest system described via \eqref{Eqn:1}. Pest eradication is not possible in finite time even if the quality of the additional food satisfies $\beta - \delta \alpha  > 0$
and the quantity of the additional food satisfies $\xi > \frac{\delta}{\beta - \delta \alpha }$.
\end{theorem}

\begin{proof}

We proceed by contradiction. Assume the following parametric restrictions on the quality and quantity of additional food are satisfied, 
 $\beta - \delta \alpha  > 0$, $\xi > \frac{\delta}{\beta - \delta \alpha }$,
and the pest $x$ goes extinct in finite time. Then

\begin{equation}
\lim_{t \rightarrow T^{*} < \infty} x(t) \rightarrow 0.
\end{equation}
Now consider the state variable $v$, defined by $v=\frac{1}{x}$. We must have
\begin{equation}
\lim_{t \rightarrow T^{*} < \infty} v(t) \rightarrow \infty.
\end{equation}
or $v$ must blow-up in finite time. 
A simple substitution $v=\frac{1}{x}$ in \eqref{Eqn:1} yields the following new system for the states $v$ and $y$.

\begin{equation}\label{Eqn:2tnan}
	\begin{split}
	&\frac{dv}{dt} = -v + \frac{1}{\gamma} + \frac{yv^{2}}{1+(1+\alpha \xi)v}    \\
    &\frac{dy}{dt}  =  \frac{\beta(1+\xi v)y}{1+ (1+\alpha \xi)v} - \delta y 
\end{split}
\end{equation}

Using positivity of the states a simple estimate yields,

\begin{equation}\label{Eqn:2tna}
	\begin{split}
	&\frac{dv}{dt} = -v + \frac{1}{\gamma} + \frac{yv^{2}}{1+(1+\alpha \xi)v} \leq -v + \frac{1}{\gamma} + yv,    \\
    &\frac{dy}{dt}  =  \frac{\beta(1+\xi v)y}{1+ (1+\alpha \xi)v} - \delta y \leq \left(\frac{\beta \xi}{1 + \alpha \xi} + \beta \right)y.
\end{split}
\end{equation}

Note, from \eqref{Eqn:1} a simple comparison with the logistic equation yields, $x < \gamma$, so 
\begin{equation}
\frac{1}{x} = v > \frac{1}{\gamma} => -v + \frac{1}{\gamma}  < 0, 
\end{equation}
inserting this in \eqref{Eqn:2tna} yields,

\begin{equation}\label{Eqn:2tnae}
	\begin{split}
	&\frac{dv}{dt} = -v + \frac{1}{\gamma} + \frac{yv^{2}}{1+(1+\alpha \xi)v}  < yv,    \\
    &\frac{dy}{dt}  =  \frac{\beta(1+\xi v)y}{1+ (1+\alpha \xi)v} - \delta y \leq \left(\frac{\beta \xi}{1 + \alpha \xi} + \beta \right)y.
\end{split}
\end{equation}

By a simple comparison argument $\tilde{v}, \tilde{y}$ are super solutions to $v,y$, where $\tilde{v}, \tilde{y}$ solve

\begin{equation}\label{Eqn:2tna1}
	\frac{d \tilde{v}}{dt} =  \tilde{y}\tilde{v},  \ \tilde{v}_{0} = v_{0} , \ \frac{d \tilde{y}}{dt}  =  \left(\frac{\beta \xi}{1 + \alpha \xi} + \beta \right) \tilde{y}, \ \tilde{y}_{0} = y_{0}.
\end{equation}

Since from \eqref{Eqn:2tna1} $\tilde{y} = y_{0}e^{\left(\frac{\beta \xi}{1 + \alpha \xi} + \beta \right)t}$, plugging this into the equation for $\tilde{v}$ yields,
 \begin{equation}
 \tilde{v}= v_{0} e^{e^{\left(\frac{\beta \xi}{1 + \alpha \xi} + \beta \right)t}}
 \end{equation}
  and cannot blow-up in finite time.
Thus by comparison 

\begin{equation}
\label{Eqn:2tna7}
v \leq \tilde{v} = v_{0} e^{e^{\left(\frac{\beta \xi}{1 + \alpha \xi} + \beta \right)t}}
 \end{equation}

and also cannot blow-up in finite time. 
This implies $x = \frac{1}{v}$ cannot go extinct in finite time, which is a contradiction to our initial assumption.

\end{proof}
For an alternate proof to theorem \ref{thm:1a} the reader is refereed to the appendix section \ref{app}.

\begin{remark}
Notice, $v$ cannot blow-up in finite time, even for arbitrary large $\xi$. This is easily seen from the form of the exponential in \eqref{Eqn:2tna7}, and taking the limit therein as $\xi \rightarrow \infty$. 
Thus even arbitrary large constant quantities of additional food, cannot drive the pest $x$ to extinction in finite time. 
The same applies for constant quality. Recall the quality of additional food is $\frac{1}{\alpha}$. Thus in order to increase the quality of the additional food one must decrease $\alpha$. However, from the form of the exponential in \eqref{Eqn:2tna7}, we see that $v$ cannot blow-up in finite time, even for arbitrary small $\alpha$, or in the limit that  $\alpha \rightarrow 0$.
\end{remark}

We next derive certain rates at which $x$ decays to the extinction state in infinite time.

\begin{lemma}
\label{lem:1a}
Consider the predator-pest system described via \eqref{Eqn:1}, and assume we remain in the region of the phase defined by 

 \begin{equation}
 y \geq x + 1 + \alpha \xi 
 \end{equation}

 then if the pest is driven to the extinction state, this occurs at the following decay rate

  \begin{equation}
x_{0}e^{-e^{\left(\frac{\beta \xi}{1 + \alpha \xi} + \beta \right)t}} \leq  x \leq \frac{x_{0}\gamma}{\gamma+x_{0} t} 
 \end{equation}

\end{lemma}

\begin{proof}
Assume $ y \geq x + 1 + \alpha \xi = \frac{1}{v} + 1 + \alpha \xi$, then we have $yv \geq 1+ (1 + \alpha \xi )v$ which implies

  \begin{equation}
-v + \frac{1}{\gamma} + \frac{yv^{2}}{1+(1+\alpha \xi)v} = -v + \frac{1}{\gamma} +\left( \frac{yv}{1+(1+\alpha \xi)v} \right) v  > -v + \frac{1}{\gamma} + v = \frac{1}{\gamma},
 \end{equation}
 
 inserting the above in \eqref{Eqn:2tnan} yields,
 
 \begin{equation}\label{Eqn:2tnael}
	\begin{split}
	&\frac{dv}{dt} = -v + \frac{1}{\gamma} + \frac{yv^{2}}{1+(1+\alpha \xi)v}  \geq  \frac{1}{\gamma},    \\
    &\frac{dy}{dt}  =  \frac{\beta(1+\xi v)y}{1+ (1+\alpha \xi)v} - \delta y \leq \left(\frac{\beta \xi}{1 + \alpha \xi} + \beta \right)y.
\end{split}
\end{equation}
 
 Integrating the above yields, 
 
  \begin{equation}\label{Eqn:2tnaen1}
	v  \geq \frac{t}{\gamma } + v_{0}
\end{equation}
 or
   \begin{equation}\label{Eqn:2tnaen2}
	x \leq \frac{x_{0}\gamma}{\gamma+x_{0} t} 
\end{equation}

The lower bound follows from the estimate via \eqref{Eqn:2tna7} where
 
   \begin{equation}\label{Eqn:2tnaen3}
	\frac{1}{x} = v  \leq \tilde{v} < v_{0} e^{e^{\left(\frac{\beta \xi}{1 + \alpha \xi} + \beta \right)t}}
\end{equation}
 
 thus
 
    \begin{equation}\label{Eqn:2tnaen4}
	x_{0} e^{-e^{\left(\frac{\beta \xi}{1 + \alpha \xi} + \beta \right)t}} \leq x,
\end{equation}
 
 and the proof is complete.
 
\end{proof}

\begin{corollary}
\label{cor:1bc2}
Consider the predator-pest system described via \eqref{Eqn:1}, and assume we remain in the region of the phase defined by 

 \begin{equation}
 y \geq x + 1 + \alpha \xi,
 \end{equation}

 then if the pest is driven to the extinction state, for an arbitrary large quantity $\xi$ of additional food, this occurs at best at the super exponential rate

  \begin{equation}
x_{0}e^{-e^{\left(\frac{\beta }{  \alpha} + \beta \right)t}} \leq  x.
 \end{equation}

Also, if the pest is driven to the extinction state, for an arbitrary high quality $\frac{1}{\alpha}$ of additional food, this occurs at best at the super exponential rate

 \begin{equation}
x_{0}e^{-e^{\left(\beta \xi + \beta \right)t}} \leq  x.
 \end{equation}

\end{corollary}

\begin{proof}
We can consider an arbitrarily large quantity of additional food by taking the limit as $\xi \rightarrow \infty$ in \eqref{Eqn:2tnaen4} to yield

 \begin{equation}
	\lim_{\xi \rightarrow \infty} x_{0} e^{-e^{\left(\frac{\beta \xi}{1 + \alpha \xi} + \beta \right)t}} =  x_{0} e^{-e^{\left(\frac{\beta }{ \alpha } + \beta \right)t}} \leq x.
\end{equation}

We can consider an arbitrarily high quality of additional food by taking the limit as $\alpha \rightarrow 0$ in \eqref{Eqn:2tnaen4} to yield

 \begin{equation}
\lim_{\alpha \rightarrow 0} x_{0} e^{-e^{\left(\frac{\beta \xi}{1 + \alpha \xi} + \beta \right)t}} = x_{0}e^{-e^{\left(\beta \xi + \beta \right)t}} \leq  x.
 \end{equation}

\end{proof}

\begin{figure}[!ht]
\includegraphics[scale=0.28]{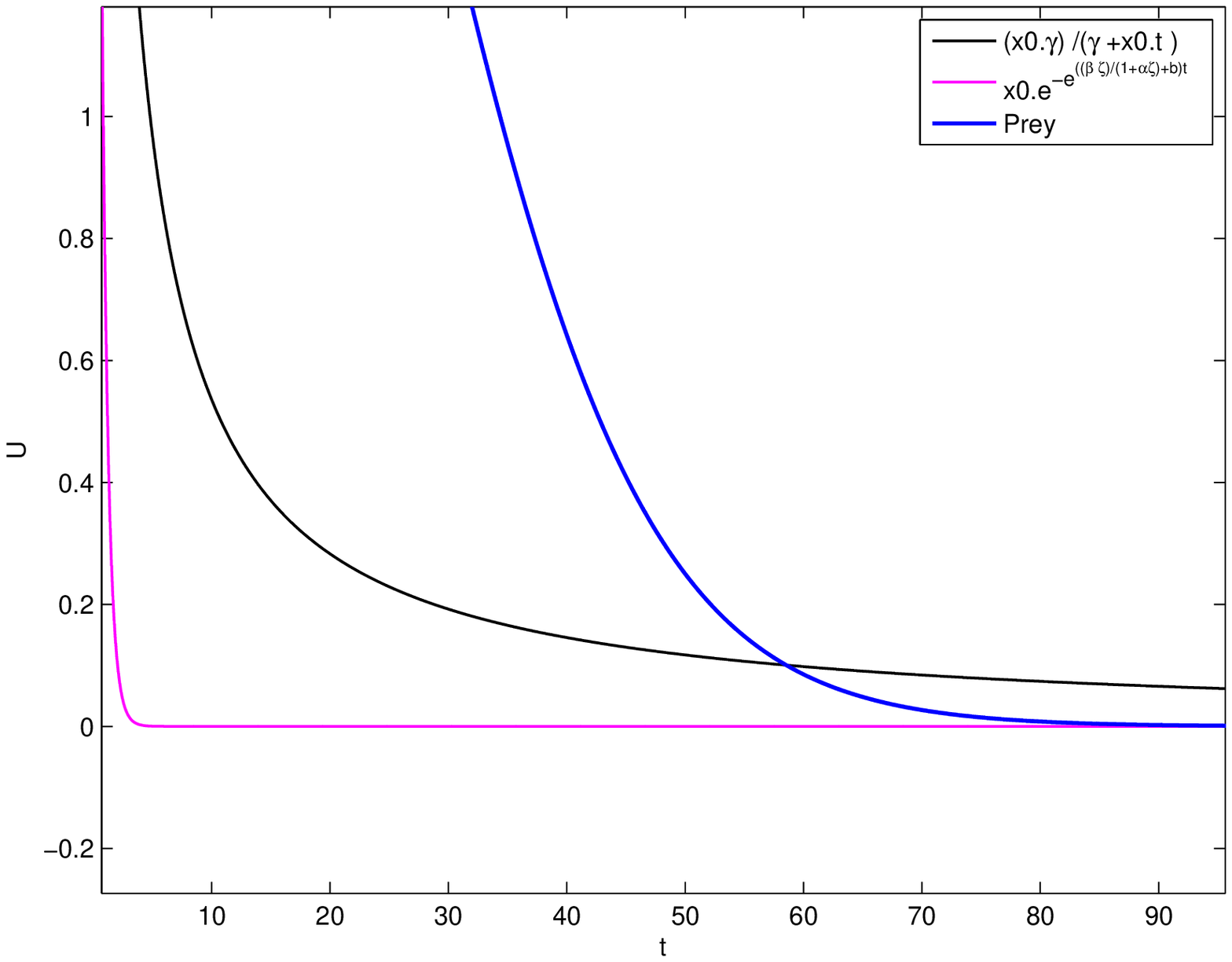}
\includegraphics[scale=0.28]{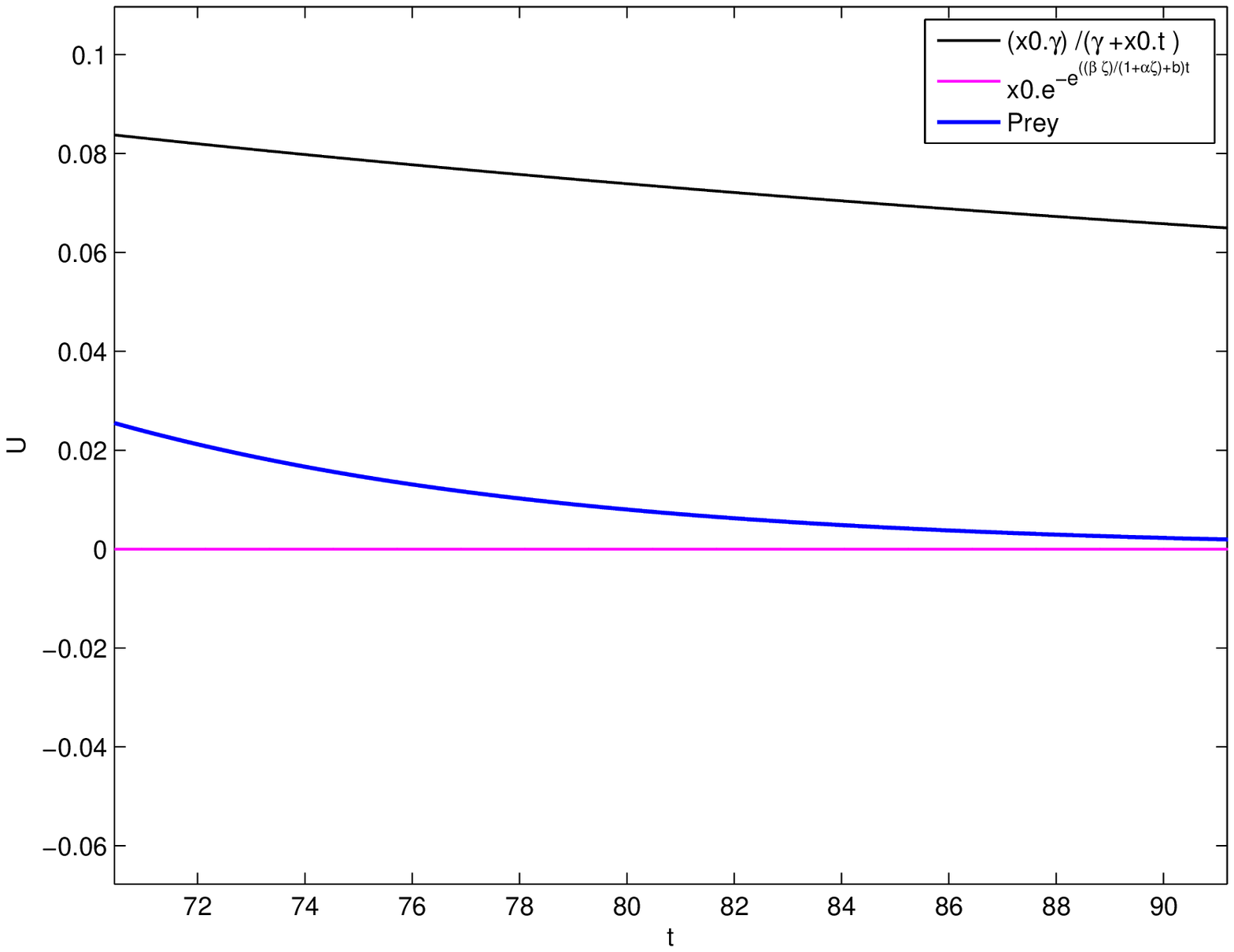}
\caption{We simulate \eqref{Eqn:1} with $\beta=0.4, \delta=0.3, \gamma = 6, \xi = 7.6, \alpha = 1.2$. Here $\xi =7.6 > 7.5 = \frac{0.3}{ 0.4 - (0.3)(1.2)} = \frac{\delta}{\beta - \delta \alpha}$. Thus according to lemma \ref{lem:1} pest eradication is possible in a finite time. We start the simulations from initial conditions $x_{0}=5.017, y_{0}=5.078$. These are compared to the polynomial and super exponential decay rates derived in lemma \ref{lem:1a}.
We see from the simulations that the decay rate of the pest to the extinction state is closer to super exponential than polynomial - this is clearly visible when we zoom in. }
\label{Fig:3sc11}
\end{figure}

\subsection{Density Dependent Quantity of Additional Food}

We show that for density dependent additional food, $v$ can blow-up in finite time, and so $x$ can go extinct in finite time. We state the following theorem,

\begin{theorem}
\label{thm:1b}
Consider the predator-pest system described via \eqref{Eqn:1}. If the quantity of the additional food $\xi$ is pest density dependent, that is 
$\xi = \xi(x)$, then for an appropriate choice of parameters and initial conditions, pest eradication is possible in finite time. 
\end{theorem}

\begin{proof}

Consider \eqref{Eqn:1} with $\xi = \xi(x)$ to yield,

\begin{equation}
\label{Eqn:2a}
\frac{dx}{dt} = x(1-\frac{x}{\gamma}) - \frac{xy}{1+\alpha \xi(x) + x}, \    \frac{dy}{dt} =\frac{\beta xy}{1+\alpha \xi(x) + x} + \frac{\beta \xi(x) y}{1+\alpha \xi(x) + x} - \delta y.
\end{equation}

We make the substitution $v=\frac{1}{x}$ in \eqref{Eqn:2a} to yield,

\begin{equation}\label{Eqn:2tnc}
	\begin{split}
	&\frac{dv}{dt} = -v + \frac{1}{\gamma} + \frac{yv^{2}}{1+(1+\alpha \xi(\frac{1}{v}))v}    \\
    &\frac{dy}{dt}  =  \frac{\beta(1+\xi(\frac{1}{v}) v)y}{1+ (1+\alpha \xi(\frac{1}{v}))v} - \delta y 
\end{split}
\end{equation}

We will proceed by constructing an appropriate $\xi$, that will cause $v$ to blow-up in finite time.
To this end we focus on the term, 
\begin{equation}
\boxed{\frac{yv^{2}}{1+(1+\alpha \xi(\frac{1}{v}))v}}
\end{equation}

 in \eqref{Eqn:2tnc}.

let us choose 
\begin{equation}
\xi\left(\frac{1}{v}\right)=\left(\frac{v^{q-1}-1}{\alpha }\right), \ 0<q<1.
\end{equation}
then

\begin{equation}
1+\left(1+\alpha\xi\left(\frac{1}{v}\right)\right)v=1+\left(1+\alpha\left(\frac{v^{q-1}-1}{\alpha }\right)\right)v=1+v^q
\end{equation}

\begin{equation}
\beta\left(1+\xi\left(\frac{1}{v}\right).v\right)y=\beta\left(1+\left(\frac{v^{q-1}-1}{\alpha }\right)v\right)y=\frac{\beta}{\alpha}\left(\alpha+v^q-v\right)y
\end{equation}
This yields the following system

\begin{equation}
\label{Eqn:v1}
    \frac{dv}{dt}=-v+\frac{1}{\gamma}+\frac{y v^{2}}{1+v^q}
\end{equation}

\begin{equation}
\label{Eqn:v21}
  \frac{dy}{dt}=\frac{\beta\left(\alpha+v^q-v\right)y}{\alpha \left(1+v^q\right)}-\delta y 
\end{equation}

Our goal is to show that there exists initial conditions and parameters s.t the $v$ solving  
\eqref{Eqn:v1} blows up in finite time.
\\
\vspace{2mm}
\\
\textbf{Case 1:} We proceed by contradiction. Assume there exists a time independent bound $M$ for $v$, that is $0 \leq v \leq M$, for any $v_{0} > 0$. Standard calculus shows the function 

\begin{equation}
f(v) = \frac{\left(\alpha+v^q-v\right)}{ \left(1+v^q\right)}
\end{equation}

is monotonically decreasing in $v$, so its minimum value is 

\begin{equation}
\frac{\left(\alpha+M^q-M\right)}{ \left(1+M^q\right)}. 
\end{equation}

In the event that this is positive that is 

\begin{equation}
\frac{\left(\alpha+M^q-M\right)}{ \left(1+M^q\right)} > \delta_{1} > 0,
\end{equation}

we choose $\beta, \alpha, \delta$ s.t,  

\begin{equation}
\frac{\beta \delta_{1}}{\alpha} - \delta > 0.
\end{equation}

In case we have,

\begin{equation}
\frac{\left(M^q-M\right)}{ \left(1+M^q\right)} <  0,
\end{equation}

we choose 

\begin{equation}
\alpha = M+1,  \ \mbox{and} \  \frac{\beta }{\alpha} - \delta > 0. 
\end{equation}

In either case

\begin{equation}
\label{Eqn:v2133}
  \frac{dy}{dt} \geq c y , \ c > 0
\end{equation}
for a positive constant $c$, where $c$ is either $ \frac{\beta }{\alpha} - \delta$ or $\frac{\beta \delta_{1}}{\alpha} - \delta$
and so $y  \geq 1$.

Thus inserting this into \eqref{Eqn:v1} we obtain,

\begin{equation}
\label{Eqn:v1nn1}
    \frac{dv}{dt} = -v+\frac{1}{\gamma}+\frac{ yv^{2}}{1+v^q} \geq -v+\frac{1}{\gamma}+\frac{ v^{2}}{1+v^q} > -v + \frac{ v^{2}}{1+v^q}
\end{equation}
However, $\tilde{v}$ solving

\begin{equation}
\label{Eqn:v3}
    \frac{d\tilde{v}}{dt}=-\tilde{v}+\frac{ \tilde{v}^{2}}{1+\tilde{v}^q}, \ \tilde{v}_{0} = v_{0}.
\end{equation}

 blows up in finite time as long as $v_{0} > (v_{0})^{q} + 1$. Thus $v \geq \tilde{v}$ by standard comparison and must also blow-up in finite time for such sufficiently large initial conditions. 
This is a contradiction and the result follows.
\\
\vspace{2mm}
\\
\textbf{Case 2:} We proceed again by contradiction. Assume now that there exists a time dependent bound for $v$, for any $v_{0} > 0$. That is WLOG say $v$ grows exponentially in time, and so 

\begin{equation}
v \leq  e^{cT}, \ t \in [0,T].
\end{equation}

As earlier, for any given $T$ we have
 the minimum value of $f(v)$ is given by
\begin{equation}
\frac{\left(\alpha+e^{qcT}-e^{cT}\right)}{ \left(1+e^{qcT}\right)}. 
\end{equation}

In the event that this is positive that is 

\begin{equation}
\frac{\left(\alpha+e^{qcT}- e^{cT}\right)}{ \left(1+e^{qcT}\right)} > \delta_{1} > 0,
\end{equation}

we choose $\beta, \alpha, \delta$ s.t,  

\begin{equation}
\frac{\beta \delta_{1}}{\alpha} - \delta > 0.
\end{equation}

In case we have,

\begin{equation}
\frac{\left(e^{qcT}-e^{cT}\right)}{ \left(1+e^{qcT}\right)} <  0,
\end{equation}

we choose 

\begin{equation}
\alpha = e^{cT}+1,  \ \mbox{and} \  \frac{\beta }{\alpha} - \delta > 0. 
\end{equation}

In either case

\begin{equation}
\label{Eqn:v2133}
  \frac{dy}{dt} \geq c y , \ c > 0
\end{equation}
for a positive constant $c$, where $c$ is either $ \frac{\beta }{\alpha} - \delta$ or $\frac{\beta \delta_{1}}{\alpha} - \delta$
and so $y  \geq 1$.

Thus inserting this into \eqref{Eqn:v1} we obtain,

\begin{equation}
\label{Eqn:v1nn}
    \frac{dv}{dt} = -v+\frac{1}{\gamma}+\frac{ yv^{2}}{1+v^q} \geq -v+\frac{1}{\gamma}+\frac{ v^{2}}{1+v^q} > -v + \frac{ v^{2}}{1+v^q}
\end{equation}

However, $\tilde{v}$ solving

\begin{equation}
\label{Eqn:v3}
    \frac{d\tilde{v}}{dt}=-\tilde{v}+\frac{ \tilde{v}^{2}}{1+\tilde{v}^q}, \ \tilde{v}_{0} = v_{0}.
\end{equation}

 blows up at a finite time $T^{*}$, as long as $v_{0} > (v_{0})^{q} + 1$. Note, for a given $T$ it is possible that $T < T^{*}$, and so the blow-up of $\tilde{v}$ at $T^{*}$, does not give us any information about $v$. However, if we choose $v_{0}$ sufficiently large, we can decrease $T^{*}$ s.t $T^{*} \leq T$, and in this case
 $v$ must also blow-up in finite time for such sufficiently large initial conditions. 
This is a contradiction and the result follows.

These results demonstrate that $v$ blows-up in finite time, hence the state variable $x=\frac{1}{v}$, must go extinct in finite time.

\end{proof}

\begin{lemma}
\label{lem:1bc}
Consider the predator-pest system described via \eqref{Eqn:2a}. If the pest density dependent quantity of additional food is non-negative,
$ \xi(x) \geq 0$, then pest eradication is not possible in finite time. 
\end{lemma}
\begin{proof}
Via the non-negativity of $\xi(x) = \xi \left(\frac{1}{v}\right)\geq 0$, and earlier estimates we have

\begin{equation}\label{Eqn:2tnae1}
	\begin{split}
	&\frac{dv}{dt} = -v + \frac{1}{\gamma} + \frac{yv^{2}}{1+(1+\alpha  \xi \left(\frac{1}{v}\right))v}  < -v + \frac{1}{\gamma} + \frac{yv^{2}}{v} <yv,    \\
    &\frac{dy}{dt}  =  \frac{\beta(1+ \xi \left(\frac{1}{v}\right) v)y}{1+ (1+\alpha  \xi \left(\frac{1}{v}\right))v} - \delta y \leq \left(\frac{\beta \xi}{1 + \alpha \xi} + \beta \right)y.
\end{split}
\end{equation}

The result follows trivially as $v < v_{0} e^{e^{\left(\frac{\beta \xi}{1 + \alpha \xi} + \beta \right)t}}$, and cannot blow-up in finite time, and so $x$ cannot go extinct in finite time.

\end{proof}

An easy consequence of the above follows,
\begin{corollary}
\label{cor:1bc1}
Consider the predator-pest system described via \eqref{Eqn:2a}. In order for pest eradication in finite time,
 the pest density dependent quantity of additional food $\xi(x)$ must change sign.

\end{corollary}

We demonstrate the results of Theorem \ref{thm:1a}, Theorem \ref{thm:1b} numerically, see Fig. \ref{Fig:3sc11}.

\begin{figure}[!ht]
\includegraphics[scale=0.3]{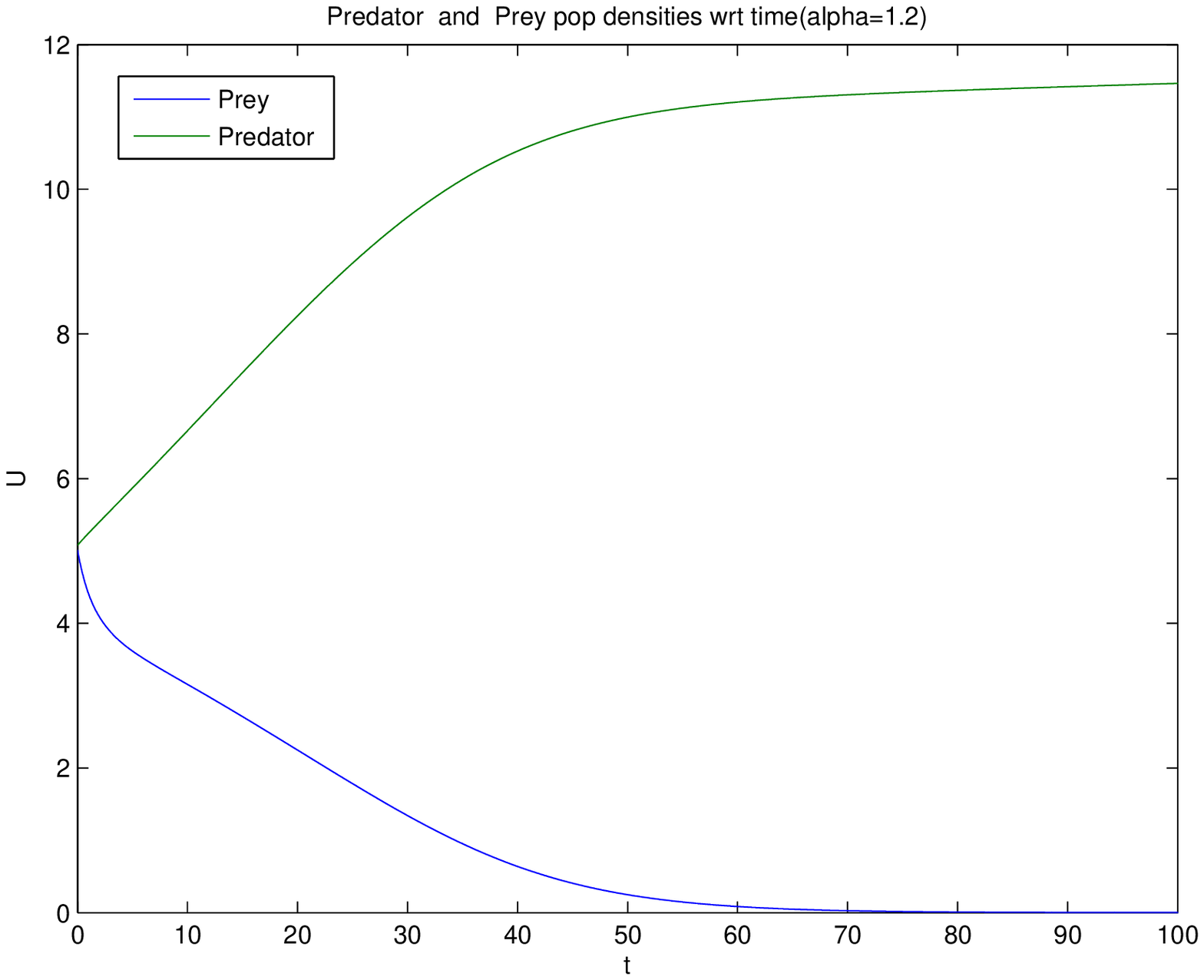}
\includegraphics[scale=0.3]{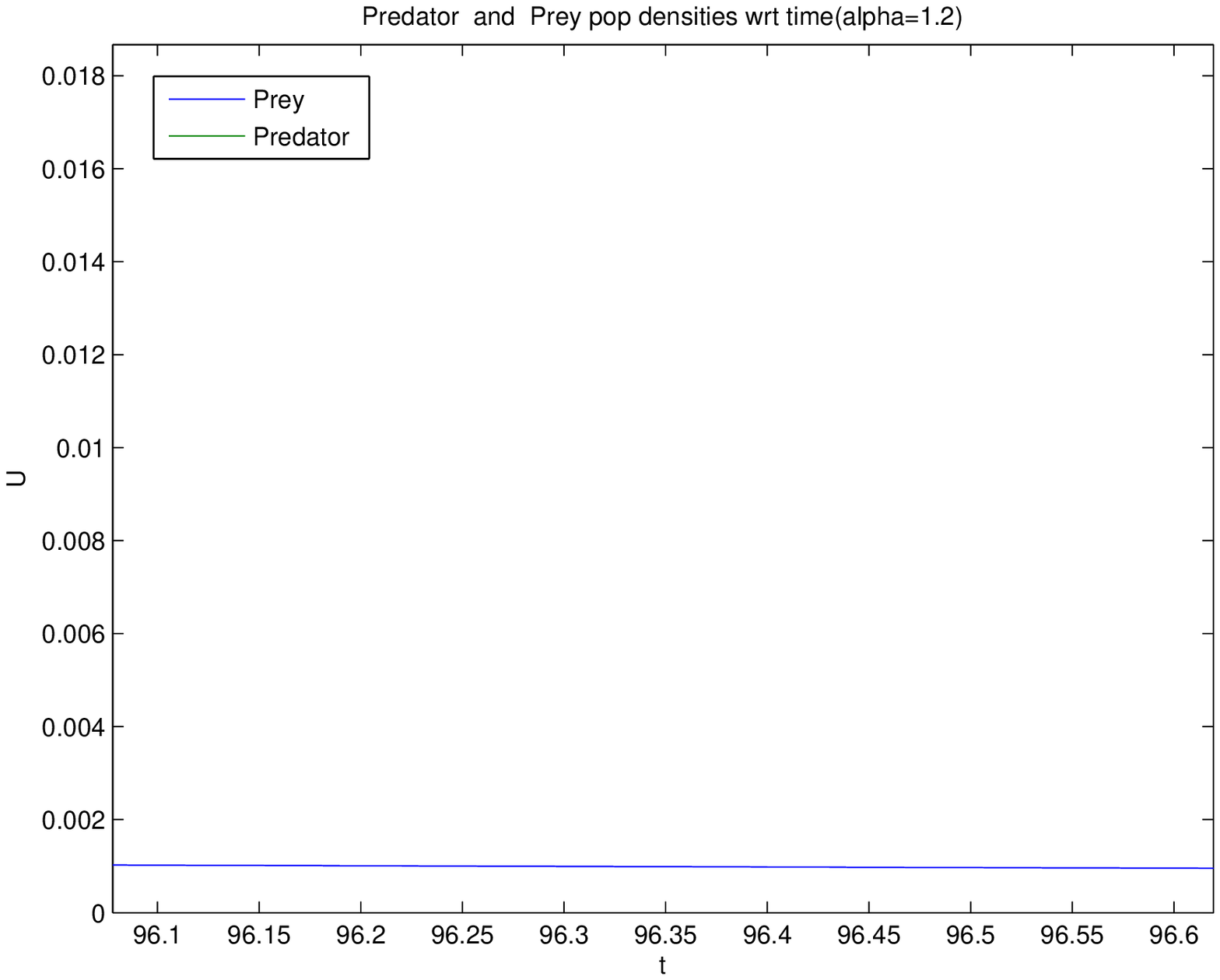}
\includegraphics[scale=0.3]{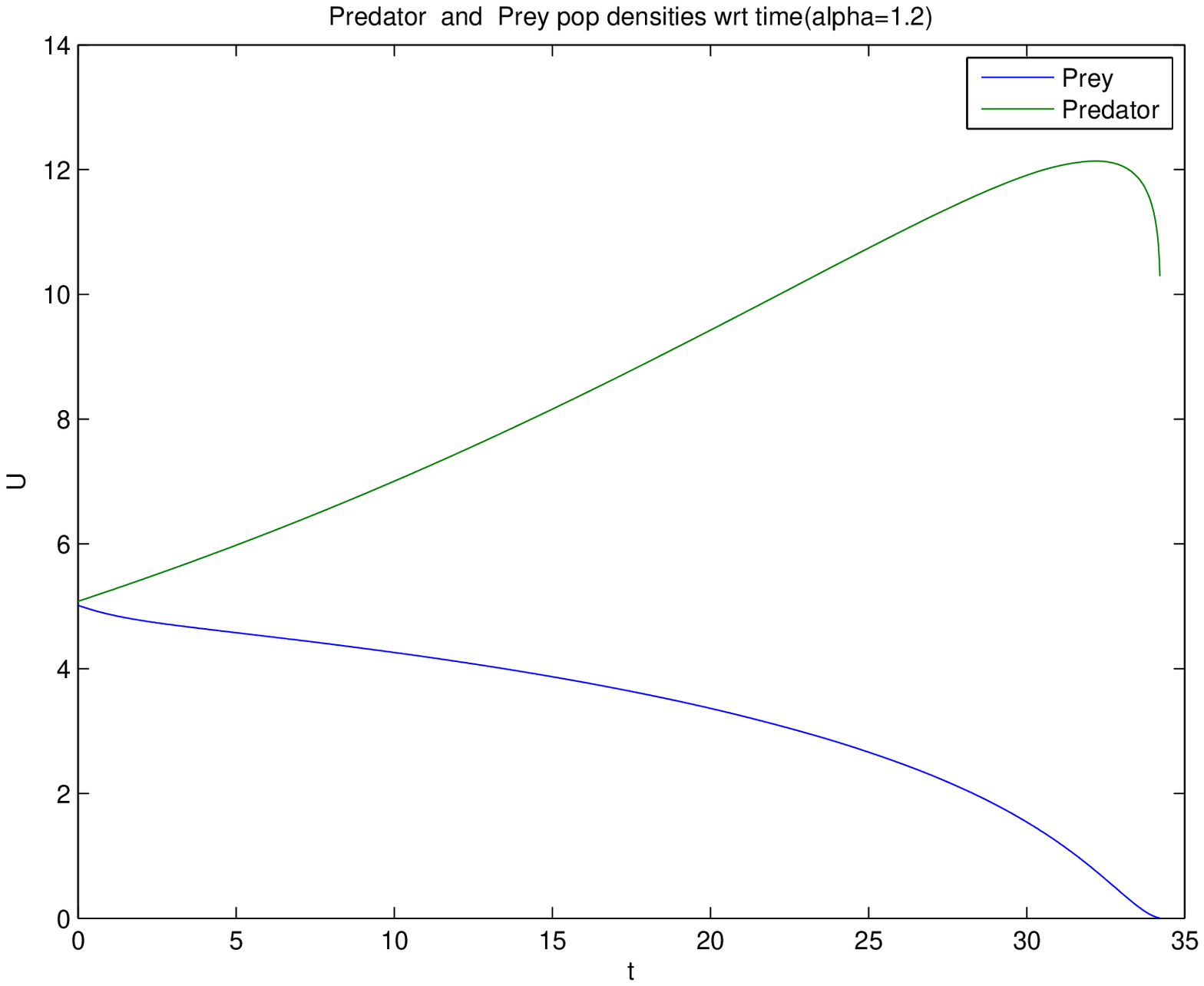}
    \caption{We simulate \eqref{Eqn:1} with $\beta=0.4, \delta=0.3, \gamma = 6, \xi = 7.6, \alpha = 1.2$. Here $\xi =7.6 > 7.5 = \frac{0.3}{ 0.4 - (0.3)(1.2)} = \frac{\delta}{\beta - \delta \alpha}$. Thus according to lemma \ref{lem:1} pest eradication is possible in a finite time. The left panel shows a simulation for initial conditions $x_{0}=5.017, y_{0}=5.078$. It seems the pest is eradicated starting at time about 70. However when we zoom into the pest density 
we see even at time 90, $x \approx 0.002$. However, when we try a density dependent introduction, $\xi(x) = 7.6 \sqrt{x} -1$, we see in the right panel finite time eradication at time, $t \approx 34.2$.}
\label{Fig:3sc11}
\end{figure}


\section{Discussion, Conjectures, and Conclusions}
\label{conj}
 Since a density dependent introduction of additional food can cause pest extinction in finite time - whereas a constant introduction cannot, we conjecture,

\begin{conjecture}
A pest density dependent introduction of additional food is more effective than a constant introduction.
\end{conjecture}

In the simulations in Fig.\ref{Fig:3sc11} we use $\xi(x)= \frac{a \sqrt{x} - 1}{\alpha}$. We insert this into \eqref{Eqn:1} to obtain the following equation for the predator,

\begin{equation}
\label{Eqn:1bd2p}
  \frac{dy}{dt} = \overbrace{\frac{\beta xy}{a \sqrt{x} + x}}^{\mbox{gain from pest}} + \overbrace{\left( \frac{\beta}{\alpha}\right)\frac{(a \sqrt{x} -1) y}{a  \sqrt{x} + x}}^{\mbox{gain from additional food}} - \delta y.
\end{equation}

Notice at low pest density $\boxed{ a \sqrt{x} -1 < 0}$. Thus we conjecture 
\begin{conjecture}
At low pest density there is negative feedback to the predator from a pest density dependent introduction of additional food.
\end{conjecture}

There are various possible explanation for the negative feedback to the predator at low pest densities. There could be \emph{intense} predator competition/interference \cite{S99, K04, B10, LB15, K10, PUS17, PB16, 1, SG01, V12} at these low densities. 
We conjecture this can be taken advantage of from a management point of view by \textbf{removing predators at low pest densities}. Thus we conjecture,

\begin{conjecture}
From a management standpoint we equate negative/positive feedback to predator removal/replenishment,
\begin{equation*}
\label{Eqn:1bc2p}
\boxed{ \overbrace{\frac{\beta ( \overbrace{a \sqrt{x} -1}^{- sign}) y}{a \sqrt{x} + x}}^{\mbox{loss to predator}}} = \mbox{predator removal},  \  \boxed{ \overbrace{\frac{\beta ( \overbrace{a \sqrt{x} -1}^{+ sign}) y}{a \sqrt{x} + x}}^{\mbox{gain to predator}}} = \mbox{predator replenishment} 
 \end{equation*}

\begin{figure}[!ht]
\includegraphics[scale=0.45]{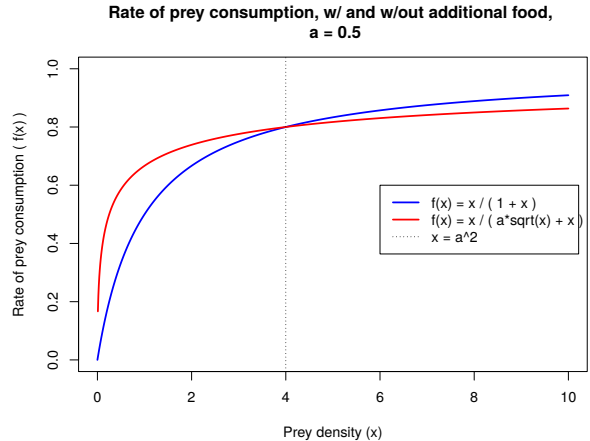}
\caption{We notice from the form of the density dependent quantity of additional food $\xi(x)$ that the predator feeding levels are higher at low pest density than if the quantity of additional food was constant. This could be a possible explanation for pest eradication in finite time via the density dependent introduction. }
\label{Fig:3sc12}
\end{figure}

\end{conjecture}

 Based on this we conjecture
 \begin{conjecture}
The effect of additional food can be mimicked by replenishing/removing the predator via a density dependent term $K_{3}(x,y)$,
\begin{equation}
\label{Eqn:12n}
\frac{dx}{dt} = x(1-\frac{x}{\gamma}) - \frac{xy}{1+ x}, \  \frac{dy}{dt} = \beta \frac{xy}{1+ x} - \delta y + K_{3}(x,y).
 \end{equation}
\end{conjecture}

We have shown for a constant quantity of additional food, pest eradication is not possible in finite time no matter how large the quantity or how high the quality of the additional food is. However, with the introduction of additional food, the pest population can be reduced to very small numbers and in a real biological population, this may turn out to be sufficient. Populations with a small numbers of individuals are often subject to increased demographic stochasticity, leading to high likelihoods of extinction \cite{PK05} despite their continuous deterministic dynamics suggesting otherwise. Future exploration of models that include stochasticity could help to determine how the likelihood of eradication changes in the presence of additional food when small population processes are considered. Alternatively, since a constant quantity of additional food cannot yield finite time eradication of the pest, some amount of stochasticity could even cause pest rebound. Proving conditions under which this might occur can also make for an interesting future direction.

In contrast, when additional food is supplied in a pest density dependent way, we show that pest eradication is possible in finite time. We suggest that the density dependent food addition affects the dynamics of this system in two key ways that drive pest population eradication in finite time. First, the presence of the secondary food source prevents the predator population from being eradicated even when their primary food source, the prey, has been reduced to a size that can no longer sustain the predator population on its own. This mechanism also applies to the constant additional food models. Secondly, and unique to density-dependent food supplementation, prey consumption rate shifts upwards at low prey density compared to prey consumption rate with constant additional food, see Fig. \ref{Fig:3sc12}. For the particular density-dependent additional food functional form we have explored here, increased consumption rate at low prey numbers is the result of the additional food quantity becoming negative. While adding a negative quantity of food is not feasible in reality, an equivalent shift might be expected if some number of predators are removed from the system. However, predator removal when predator numbers are already very low may also prove to be difficult or impossible in a real system. Note, the apparent loss to the predator coming from the loss term $\frac{\beta ( \overbrace{a \sqrt{x} -1}^{- sign}) y}{a \sqrt{x} + x}$, at low pest density (when $\sqrt{x} < \frac{1}{a}$), might not mean predator removal - rather it could be related to something that the predator has to give up as a means to increase consumption at low pest densities. Investigating possible explanations for this would make for interesting future work.

Regardless of the mode by which this is applied in real systems, a key to pest eradication seems clear - increase per capita prey consumption rates when when the prey population size is low. Exploration of other functional forms for the density-dependent food supplementation may prove fruitful in this regard. In a real system, shifts to the per capita prey consumption rate might also come through changes in predator behaviour, for example shifts in prey handling time. 

Mathematical models of biocontrol can provide an idea of what dynamics are possible and suggest routes by which pest eradication is theoretically feasible. Future directions involve studying the effects of pest refuge, evolutionary effects as well as stochastic effects \cite{PQB16, PK05, B13, B15, F14, H93}. However, experimental tests are required to assess the biological reality of applying these strategies. Laboratory experiments using dynamically interacting predator (protozoa), prey (bacteria), and additional food (a chemical supplement) are currently underway and will help to provide an additional intermediate step linking theory to successful biocontrol applications in the natural world.
\section{Acknowledgements}
RP would like to acknowledge valuable support from the National Science Foundation via awards DMS 1715377 and DMS 1839993.

\section{Appendix}
\label{app}
We next provide an alternate proof to Theorem \ref{thm:1a}

\begin{proof}
Consider the new variable $V = \frac{y}{x}$, then we obtain the following equation for the dynamics of $V$,

\begin{equation}\label{Eqn:2tna3}
	\frac{d V}{dt} = \left( \frac{\beta x}{1+\alpha \xi + x} + \frac{\beta \xi }{1+\alpha \xi + x} - \delta  - (1-\frac{x}{\gamma})\right) V + \left (  \frac{ x}{1+\alpha \xi + x}  \right)V^{2}
\end{equation}

Clearly via positivity of the states we have,

\begin{eqnarray}\label{Eqn:2tna3}
	&& \frac{d V}{dt} \nonumber \\
	&=& \left( \frac{\beta x}{1+\alpha \xi + x} + \frac{\beta \xi }{1+\alpha \xi + x} - \delta  - (1-\frac{x}{\gamma})\right) V + \left (  \frac{ x}{1+\alpha \xi + x}  \right)V^{2}  \nonumber \\
	&<&  C x V^{2}  \nonumber \\
\end{eqnarray}

for some constant $C$. Thus consider a super solution  $\tilde{V} =  \frac{y}{x}$ to $V$, which solves

\begin{equation}\label{Eqn:2tna35}
	\frac{d \tilde{V}}{dt} =  C x  \tilde{V}^{2}, \ \tilde{V}_{0} = V_{0}.
	\end{equation}
	
We now argue by contradiction.
Assume $\tilde{V}$ blows-up at the finite time $T^{*}$, this is only possible if $x$ goes extinct at the same finite time $T^{*}$, as $y$ is bounded by $e^{\left(\frac{\beta \xi}{1 + \alpha \xi} + \beta \right)t}$, that is,

\begin{equation}
\lim_{t \rightarrow T^{*} < \infty} x(t) \rightarrow 0, \ \lim_{t \rightarrow T^{*} < \infty} V(t) \rightarrow \infty.
\end{equation}

We integrate \eqref{Eqn:2tna35} in the time interval $[T^{*}-\delta, T^{*}]$, ($\delta <<1$), to obtain

\begin{equation}
\label{Eqn:2tna3}
\infty = 	 \tilde{V}(T^{*}) =  \frac{1}{  \frac{1}{ \tilde{V}(T^{*}-\delta)}  - \int^{T^{*}}_{T^{*}-\delta}x(s)ds    }
\end{equation}

the only way we can have equality in \eqref{Eqn:2tna3} is if

\begin{equation}\label{Eqn:2tna31}
 \frac{1}{ \tilde{V}(T^{*}-\delta)}  - \int^{T^{*}}_{T^{*}-\delta}x(s)ds  = 0
\end{equation}

However by the continuity of the state variables $ \tilde{V}, x$, upto (but not including the blow-up/extinction time $T^{*}$) we have that for any $\epsilon > 0$ (assuming $\epsilon << 1$) there exists a $\delta(\epsilon) > 0$ s.t 

\begin{equation}\label{Eqn:2tna310}
 \frac{1}{\epsilon} < \tilde{V}(T^{*}-\delta) <   K(\delta),
\end{equation}
Where $K$ is a function of $\delta$, and
\begin{equation}\label{Eqn:2tna312}
  \int^{T^{*}}_{T^{*}-\delta}x(s)ds \leq |T^{*}-(T^{*}-\delta)||x|  < \delta C\epsilon
\end{equation}

Here $C$ is an upper estimate on $y$ at time $T^{*}$, so $C=e^{e^{\left(\frac{\beta \xi}{1 + \alpha \xi} + \beta \right)T^{*}}}$, (as $\tilde{V} = \frac{y}{x}$).

Thus
\begin{equation}\label{Eqn:2tna31}
 \frac{1}{ \tilde{V}(T^{*}-\delta)}  - \int^{T^{*}}_{T^{*}-\delta}x(s)ds  >  \frac{1}{K(\delta)} - C \delta \epsilon,
\end{equation}

choosing $\epsilon = \frac{\epsilon}{C}$, we see that as long as we choose $K(\delta) \delta < \frac{1}{\epsilon}$, we have,

\begin{equation}\label{Eqn:2tna31}
 \frac{1}{ \tilde{V}(T^{*}-\delta)}  - \int^{T^{*}}_{T^{*}-\delta}x(s)ds  >  \frac{1}{K(\delta)} - \delta \epsilon > 0,
\end{equation}



Thus we have a contradiction to \eqref{Eqn:2tna3}, and $\tilde{V}$ cannot have blown up at $T^{*}$. Since $T^{*}$ is arbitrary we can conclude $\tilde{V}$ cannot blow up in finite time. Since $V < \tilde{V}$, $V$ also cannot blow up in finite time by comparison. Thus $x$ cannot go extinct in finite time.

\end{proof}

\end{document}